%% file: main.tex
\documentclass[11pt]{article}
\usepackage[utf8]{inputenc}
\usepackage{amsfonts,latexsym,amsthm,amssymb,amsmath,amscd,euscript,mathrsfs,graphicx}
\usepackage{framed}
\usepackage{fullpage}
\usepackage{color}
\usepackage{enumitem}
\usepackage{algorithm}
\usepackage{algorithmic}
\usepackage[obeyFinal,textsize=scriptsize,shadow]{todonotes}

\DeclareMathOperator{\disj}{\mathsf{MostlyDISJ}}

\newtheorem{theorem}{Theorem}
\newtheorem{proposition}{Proposition}
\newtheorem{lemma}{Lemma}
\newtheorem{fact}{Fact}
\newtheorem{corollary}{Corollary}

\theoremstyle{definition}
\newtheorem{Definition}{Definition}[section]

\theoremstyle{remark}
\newtheorem*{Remark}{Remark}

\newcommand{\BE}{\mathbb E}

\newcommand{\BP}{\mathbb P}

\newcommand{\BR}{\mathbb R}

\newcommand{\eps}{\varepsilon}

\title{Frequency Estimation with One-Sided Error}

\author{Piotr Indyk\thanks{MIT. Email: \texttt{indyk@mit.edu}}
\and Shyam Narayanan\thanks{MIT. Email: \texttt{shyamsn@mit.edu}}
\and David P. Woodruff\thanks{CMU. Email: \texttt{dwoodruf@cs.cmu.edu}}}

\begin{document}

\begin{titlepage}
\maketitle
\thispagestyle{empty}

\begin{abstract}
    Frequency estimation, also known as the Point Query problem, is one of the most fundamental problems in streaming algorithms.  Given a stream $S$ of elements from some universe $U=\{1 \ldots n\}$, the goal is to compute, in a single pass,  a short ``sketch'' of $S$ so that for any element $i \in U$, one can estimate the number $x_i$ of times $i$ occurs in $S$ based on the sketch alone. Two state of the art solutions to this problems are Count-Min and Count-Sketch algorithms. They are based on linear sketches, which means that the data elements can be deleted as well as inserted and sketches for two different streams can be combined via addition. However, the guarantees offered by Count-Min and Count-Sketch are incomparable. The frequency estimator $\tilde{x}$ produced by Count-Min sketch, using $O(1/\varepsilon \cdot \log n)$ dimensions, guarantees that (i) $\|\tilde{x} -x\|_{\infty} \le \varepsilon \|x\|_1$ with high probability, and (ii) $\tilde{x} \ge x$ holds deterministically. Also, Count-Min works under the assumption that $x \ge 0$. On the other hand, Count-Sketch, using $O(1/\varepsilon^2 \cdot \log n)$ dimensions, guarantees that $\|\tilde{x} -x\|_{\infty} \le \varepsilon \|x\|_2$ with high probability. A natural question is whether it is possible to design the ``best of both worlds'' sketching method, with error guarantees depending on the $\ell_2$ norm and space comparable to Count-Sketch, but (like Count-Min) also has the no-underestimation property.

    Our main set of results shows that the answer to the above question is negative. We show this in two incomparable computational models: linear sketching and streaming algorithms. Specifically, we show that:
    \begin{itemize}
        \item Any {\em linear sketch} satisfying the $\ell_p$ norm error guarantee  with probability at least $2/3$ and having the no-underestimation property must be of dimension of at least $\Omega(n^{1-1/p} /\varepsilon)$, even if the sketched vectors are non-negative. This bound is tight, as we also give a linear sketch of dimension $O(n^{1-1/p}/\varepsilon)$ satisfying these properties.
        \item Any {\em streaming algorithm}  satisfying the $\ell_p$ norm error guarantee  with probability at least $2/3$ and having the no-underestimation property must use at least  $\Omega(n^{1-1/p} /\varepsilon)$ bits. This holds even for algorithms that only allow insertions and make any constant number of passes over the stream. This bound is tight up to a logarithmic factor.  
\end{itemize}    
   
    We also study the complementary problem, where the sketch is required to not {\em over}-estimate, i.e., $\tilde{x} \le x$ should hold always. We show that any linear sketch satisfying this property and having the $\ell_p$ error guarantee with probability at least $2/3$ must be of dimension at least $\Omega(n^{1-1/p}/\varepsilon )$. We also show that this bound is tight up to polylogarithmic factors, by providing an appropriate linear sketch.
\end{abstract}

\end{titlepage}

\section{Introduction}
\input{intro}


\subsection{Overview of Techniques}
\input{overview}

\section{No-Underestimation Sketching Lower Bound}

\begin{theorem} \label{mainNoUnder}
    Fix $1 \le T \le n.$ Let $v = \left(\frac{1}{n}, \frac{1}{n}, \dots, \frac{1}{n}\right) \in \BR^n$. Then, there exists an absolute constant $c > 0$ such that for any $k \le c \cdot n/T$ and any real-valued matrix $A \in \BR^{k \times n}$, there exists some $x \in \BR^n$ such that $Ax = Av,$ $x$ has only nonnegative entries, and $\max x_i \ge T/n.$
\end{theorem}


First, we note the following standard fact. We include a proof in the appendix for completeness. 

\begin{lemma} \label{RankLB1}
    For any square (possibly non-symmetric) matrix $R$, $rk(R) \gtrsim tr(R)^2/||R||_F^2.$
\end{lemma}
We next need a key matrix inequality. We recently learned that the following result is known, and follows directly from Theorem 1.3 of \cite{fan1954rankmatrix}. We provide an independent proof in the appendix.

\begin{lemma} \label{RankLB2}
    Let $1 \le T \le n$, and suppose that $M \in \BR^{n \times n}$ is such that for all $1 \le i \le n,$ $|M_{ii}| \ge 1$ and $\sum_{1 \le j \le n} |M_{ij}| \le T$. Then, the rank $k$ of $M$ satisfies $k = \Omega(n/T)$.
\end{lemma}

We now prove Theorem \ref{mainNoUnder}.

\begin{proof}
For each $1 \le i \le n$, consider the linear program $\max x_i: Ax = Av, x \ge 0$ where $A \in \BR^{k \times n}$ and $x, v \in \BR^{n}.$ Writing $x_i = e_i^T x$ for $e_i$ the $i$th unit vector, this program's dual is $\min_y v^T A^T y: A^T y \ge e_i$. If we assume the theorem is false for some fixed $A$, then for all $i$, there is some row vector $z^{(i)} = y^T A \ge e_i$ (coordinate-wise) such that $z^{(i)} \in RowSpan(A)$ (since $z^{(i)} = y^T A$) and $\sum_j z^{(i)}_j \le T$ (since $v^T A^T y \le T/n$). Therefore, letting $M$ be the matrix such that $M$'s $i$th row is $z^{(i)}$, we have that $M_{ii} \ge 1$ for all $i$, $M_{ij} \ge 0$ for all $j$, and $\sum_{j} |M_{ij}| \le T.$ Finally, all of $M$'s rows are in the rowspan of $A$, so $M$'s rank is at most $k$.

Therefore, by Lemma \ref{RankLB2}, we  have that $k = \Omega(n/T)$.
\end{proof}


\begin{corollary}
\label{cor:noUnder}
    Any linear sketch that returns $\hat{x}$ such that $\hat{x}_i \ge x_i$ holds deterministically and $\hat{x}_i \le x_i + \varepsilon \cdot ||x||_p$ {for all $i \in U$} holds with probability at least $2/3$ must use  $\Omega(\min(n, \varepsilon^{-1} \cdot n^{1 - 1/p}))$ rows.
\end{corollary}

\begin{proof}
    Suppose $A$ is a randomized sketch matrix and the stream has vector $v = (\frac{1}{n}, \dots, \frac{1}{n})$. The sketch is $Av$, and the output $\hat{v}$ must satisfy $\hat{v}_i \ge x_i$ for all nonnegative vectors $x$ with $Ax = Av$ by our no underestimation assumption. Therefore, by Theorem \ref{mainNoUnder}, if $A$ has $k$ rows, if we choose $T = c \cdot n/k,$ then $\max \hat{v}_i \ge T/n = c/k.$ However, we are assuming that the sketching algorithm returns $\hat{v}_i \le v_i + \varepsilon \cdot ||v||_p = 1/n + \varepsilon/n^{1-1/p}$. Thus, $c/k \le 1/n + \varepsilon/n^{1-1/p},$ so $k \ge \Omega(\min(n, \varepsilon^{-1} \cdot n^{1-1/p}))$.
\end{proof}

\section{Multi-pass Insertion-Only Stream Space Lower Bound}
%
In the multiparty communication model we consider $k$-ary 
functions $F: \mathcal{L} \rightarrow\mathcal{Z}$ where $\mathcal{L} \subseteq 
\mathcal{X}_1 \times \mathcal{X}_2 \times \cdots \times \mathcal{X}_k$. There 
are $k$ players who receive inputs $X_1, \dotsc, X_k$, respectively. 
We consider protocols 
in the blackboard model where in a protocol $\pi$, the players speak in any order and possibly
multiple times, 
and the player who speaks next is determined by the protocol transcript.  
Each player's message is posted on a blackboard and is seen by all other players. 
A message of Player $i$ is a function of the messages on the blackboard thus far, $X_i$, 
and the private randomness of Player $i$. 
The $k$-th player's 
final message is the output of the protocol.
The communication cost of a multiparty protocol $\pi$ is the sum of the lengths 
of all individual messages. 
A protocol $\pi$ is a $\delta$-error protocol for the function $f$ if for every 
input $x\in \mathcal{L}$, the output of the protocol equals $f(x)$ with 
probability at least 
$1-\delta$. 
The randomized communication complexity $R_{\delta}(f)$ of $f$ 
is the cost of the cheapest randomized protocol that computes 
$f$ correctly on every input with error probability at most $\delta$, where the
probability is taken only over the private randomness of the players. 

For background on information complexity, see, e.g., \cite{bar2004information,bar2002complexity}.
Let $H(X)$ denote the Shannon entropy of the random variable $X$.
Let $H(X\ |\ Y)$ denote the conditional entropy of $X$ given $Y$. Let $I(X; Y) = H(X) - H(X\ |\ Y)$ denote the mutual information 
and $I(X; Y\ |\ Z)$ denote the conditional mutual information, for random variables $X, Y,$ and $Z$. 
\begin{proposition}\label{prop:mut}
Let $X, Y, Z, W$ be random variables.
\begin{enumerate}
\item If $X$ takes value in $\{1,2, \ldots, m\}$, then $H(X) \in [0, \log_2 m]$.

\item $H(X) \ge H(X\ |\ Y)$ and $I(X; Y)  = H(X) - H(X\ |\ Y) \ge 0$.

\item  For any random variables $X, Y, W, Z$, we have 
$I(X; Y\ |\ W) \le I(X; Y\ |\ Z, W) + H(Z)$. 

\item For any random variables $X_1, X_2, \ldots, X_n, Y$,
$\textstyle I(X_1, \ldots, X_n; Y) = \sum_{i = 1}^n I(X_i; Y\ |\ X_1, \ldots, X_{i-1})$.
\end{enumerate}
\end{proposition}
We note that part 3 of Proposition \ref{prop:mut} follows from the fact that, using the chain
rule for mutual information and expanding $I((X, Z); Y | W)$ in two different ways:
$$I(X ; Y\ |\ W) + I(Z ; Y\ |\ X, W)
= I(Z ; Y\ |\ W) + I(X ; Y\ |\ Z, W),$$
which then implies
$$I(X; Y\ |\ W) + H(Z\ |\ X, W) - H(Z\ |\ X, Y, W) = I(X; Y\ |\ Z, W) + H(Z\ |\ W) - H(Z\ |\ Y, W),$$
and so 
$I(X ; Y\ |\ W) = I(X ; Y\ |\ Z, W) + H(Z \ |\ X, Y, W) - H(Z \ |\ X, W) - H(Z \ |\ Y, W) + H(Z\ |\ W)$,
and then using that $H(Z \ |\ X, Y, W) \leq H(Z \ |\ X, W)$ and $H(Z\ |\ W) \leq H(Z)$ by Part 2, and
that $H(Z \ |\ Y, W) \geq 0$. 

\begin{fact}\label{fact:geometric} (Example 4.6.1 of \cite{leon1994probability})
Let $R$ be a geometric random variable with success probability $p$. Then
$H(R) = \log_2 \left (\frac{1}{p} \right ) + \frac{1-p}{p} \cdot \log_2 \left (\frac{1}{1-p} \right) = \frac{H(p)}{p}$, where $H(p) = p \log_2 \left (\frac{1}{p} \right) + (1-p) \log_2 \left (\frac{1}{1-p} \right )$. 
\end{fact}
\begin{Definition}
	Let $\pi$ be a randomized protocol whose inputs belong to 
	$\mathcal{K}\subseteq \mathcal{X}_1 \times \mathcal{X}_2 \dotsc \times 
	\mathcal{X}_k$. Suppose $((X_1, X_2, \dotsc, X_k), D)\sim \eta$ where 
	$\eta$ is a distribution over $\mathcal{K}\times\mathcal{D}$ for some 
	set $\mathcal{D}$. The conditional information cost of $\pi$ 
	with respect to $\eta$ is defined as:
	$cCost_{\eta}(\pi) = I(X_1, \dotsc, X_k; \pi(X_1, \dotsc, X_k) \mid D).$
	Here $\pi(X_1, \dotsc, X_k)$ denotes the
	transcript of the protocol $\pi$. 
	
\end{Definition}
\begin{Definition}
	The $\delta$-error conditional information complexity $CIC_{\eta, \delta}(f)$ of $f$, 
	with respect to $\eta$, is  
	the 
	minimum conditional information cost of a $\delta$-error protocol for 
	$f$ with respect to $\eta$.  
\end{Definition}
Note that in the definition of $CIC_{\eta, \delta}(f)$, the 
protocol is correct on every input with failure probability
	at most $\delta$, where the probability is taken only over the private coins of the players, 
	while we measure the information of the protocol with respect to the distribution $\eta$. 

\begin{fact} (See, e.g., Corollary 4.7 of \cite{bar2004information})
For all distributions $\eta$, we have $R_{\delta}(f) \geq CIC_{\eta, \delta}(f)$.  
\end{fact}

\begin{Definition}
	Denote by $\disj_{n, l, k}$, the multiparty Mostly 
	Set-Disjointness problem in which for each $j \in [k]$, Player $j$ receives an 
	$n$-dimensional input vector $X_j = (X_{j,1}, \dotsc, X_{j,n})$ where 
	$X_{j,i}\in \{ 0,1\}$ and the input to the protocol falls into either of 
	the following cases:
	\begin{itemize}
		\item \textbf{NO:} For all $i\in[n]$, $\sum_{j\in [k]} 
		X_{j,i}\leq 1$
		\item \textbf{YES:} There is a unique $i\in[n]$ with 
		$\sum_{j\in [k]} X_{j,i} = l$, and for all other $i'\neq i,  
		\sum_{j\in [k]} X_{j,i'}\leq 1$
	\end{itemize}
        Player $k$ must output $1$ if the input is in the YES case and $0$ in the NO case. 
\end{Definition}

Let $\mathcal{L}\subset \{0, 1\}^k$ 
be the set of elements in $x \in 
\{0, 1\}^k$ with $\sum_{j\in [k]} x_j \leq 1$ or $\sum_{j\in [k]} x_j = l$. 
Let $\mathcal{L}_n\subset \mathcal{L}^n$ denote the set of valid inputs to the 
$\disj_{n,l,k}$ function. 
Define distribution $\eta$ over $\mathcal{L}_n\times [k]^n$: 
 for each $i\in [n]$ pick $D_i\in[k]$ uniformly at random and sample
	$X_{D_i,i}$ uniformly from $\{0 ,1\}$  and for all $j'\neq D_i$ set 
	$X_{j', i}=0$. 
Let $\mu_0$ be the distribution for a given $i\in [n]$.  
Let $\eta_0 = \mu_0^n$.

\begin{theorem}(See the one-line proof of Theorem 3.8 of \cite{KPW21}
): For any $0 < \delta, c < 1$ and $2 \leq k 
	\le \frac{\log(\frac{1}{2e\delta})}{c \log(e/c)}$, 
$CIC_{\eta_0, \delta}(\disj_{n, ck, k}) = \Omega(n(1-c)^2).$
\end{theorem}
\begin{corollary}\label{cor:main}
For $c = e/4$ and $0 \leq \delta \leq e^{-ke/2}/(2e)$,
$CIC_{\eta_0, \delta} (\disj_{n, ck, k}) = \Omega(n).$ 
\end{corollary}
The above holds for communication protocols with any number of rounds
of communication.\\ 

\noindent {\bf Main Lower Bound:} Let $p \geq 1$ be any positive real number, which we assume is a constant independent of $n$. We now prove our $\Omega(n^{1-1/p}/\varepsilon)$ bits of space lower bound for any randomized $O(1)$-pass streaming algorithm in the insertion-only model, which, given a vector $x \in \{0, 1, 2, \ldots, M\}^n$ of insertions to its coordinates, outputs a vector $\tilde{x}$ 
with the following two properties:
    (1) With probability at least $2/3$, $\|\tilde{x} - x\|_{\infty} \leq \varepsilon \|x\|_p$, and  
    (2) with probability $1$, we have $\tilde{x} \geq x$. 
If a streaming algorithm $S$ satisfies the above two properties we call it a {\it no-underestimation $\ell_p$-point query algorithm}. For the lower bound it will suffice for $M$ to be less than $n$. 
Let $S$ be a no-underestimation $\ell_p$-point query algorithm. We construct a protocol $\Pi_S$ to solve $\disj_{n, ck, k}$, where
$c = e/4$ and $k = 4 \varepsilon n^{1/p}$. The protocol is described in Algorithm \ref{fig:1}. 

\begin{algorithm}[H]
	\caption{Construction of Protocol $\Pi_S$ from a no-underestimation $\ell_p$-point query algorithm $S$}
	\label{fig:1}
	{\bfseries Input}: For $j \in [k]$, Player $j$ receives the input $X_j$ to the 
	$\disj_{n, ck, k}$ problem, where $c = e/4$ and $k = 4 \varepsilon n^{1/p}$. Let $S$ be a 
	no-underestimation $\ell_p$-point query algorithm.\\
	{\bfseries Output}: Player $k$ declares whether the input to $\disj_{n, ck, k}$ is a YES or NO instance.\\
	{\bfseries Procedure}:
	\begin{algorithmic}[1]
		\STATE For $j = 1, \ldots, k$, Player $j$ creates a stream $T^j$ of insertions of items $i$ for which $X_{j,i} = 1$. Player $j$ computes $S(T^1 \circ \cdots \circ T^j)$, 
		where $T^1 \circ \cdots \circ T^j$ is the concatenation of the first $j$ streams, and Player $j$  posts the state of the streaming algorithm to the blackboard. 
		\STATE Player $k$ computes the output $\tilde{x}$ of $S(T^1 \circ \cdots \circ T^k)$. Let $I = \{i : \tilde{x}_i \geq c k\}$. 
		\STATE If $I = \emptyset$, Player $k$ terminates the protocol and outputs ``NO instance". 
		\STATE Else if $|I| = \{i\}$ for an $i \in [n]$, Player $k$ posts $i$ to the blackboard, and for $j = 1, \ldots, k$, Player $j$ posts $X_{j,i}$ to the blackboard. If $x_i = ck$, Player $k$ terminates the protocol and outputs ``YES instance". Else $x_i \in \{0,1\}$, and Player $k$ terminates the protocol and outputs ``NO instance". 
		\STATE Else $|I| > 1$, and Player $k$ writes ``start over" on the blackboard. Goto Step 1.
		\end{algorithmic}
\end{algorithm}

\begin{lemma}\label{lem:correctness} (Always Correct)
For any $1/(e n^{1/p}) < \varepsilon < 1$, given an $O(1)$-pass no-underestimation $\ell_p$-point query algorithm $S$, the protocol $\Pi_S$
solves $\disj_{n, ck, k}$ with probability $1$, where $c = e/4$ and $k = 4 \varepsilon n^{1/p}$.  
\end{lemma}
\begin{proof}
Let $I = \{i : \tilde{x}_i \geq ck\}$ be the value of the set $I$ when Player $k$ outputs and terminates the protocol. Then it must be that $|I| = 0$ or $|I| = 1$. 

In a YES instance, we have $x_i = ck$ for a unique value of $i$, and since $\tilde{x} \geq x$ for a no-underestimation algorithm $S$ with probability $1$, we have $\tilde{x}_i \geq ck$ and thus $i \in I$. Consequently, $I = \{i\}$, and
$X_i = ck$, and so Player $k$ will terminate and output ``YES instance" in Step 4. 

In a NO instance, if $\|\tilde{x}\|_{\infty} < ck$, then $I = \emptyset$ and Player $k$ will terminate
and output ``NO instance" in Step 3. Otherwise, $\|\tilde{x}\|_{\infty} \geq ck$ and so $|I| = \{i\}$
for an $i \in [n]$. Since we are in a NO instance, $x_i \in \{0,1\}$, and Player $k$ terminates
the protocol and outputs ``NO instance" in Step 4. This assumes that $ck > 1$, that is, that
$e \cdot \varepsilon n^{1/p} > 1,$ or equivalently, $\varepsilon > 1/(e n^{1/p}).$

Thus, for both YES and NO instances, $\Pi_S$ solves
$\disj_{n, ck, k}$ with probability $1$. 
\end{proof}


\begin{lemma}\label{lem:cic} (Conditional Information Cost)
For any $\varepsilon$ with $1/((e-2) n^{1/p}) \leq \varepsilon < 1/4$, given an $O(1)$-pass no-underestimation $\ell_p$-point query algorithm $S$ with space $s$ and parameter $\varepsilon$, 
and for the distribution 
$\eta_0$ defined earlier, the protocol $\Pi_S$ satisfies the following
conditional information cost bound: $cCost_{\eta_0}(\Pi_S) = O(s k),$
where $k = 4 \varepsilon n^{1/p}$. 
\end{lemma}

\begin{proof}
First note that regardless of whether we are in a YES instance or a NO instance 
of $\disj_{n, ck, k}$, by the triangle
inequality, 
\begin{eqnarray}\label{eqn:pNorm}
\|x\|_p \leq n^{1/p} + k = n^{1/p} + 4 \varepsilon n^{1/p} < 2n^{1/p}, 
\end{eqnarray}
where we have used $\varepsilon < 1/4$. 

Let $R \geq 1$ be an integer random variable indicating the total number of times Step 1 is executed. 
In order for Step 1 to be executed again after it has completed, we must have $|I| > 1$, where
$I = \{i : \tilde{x}_i \geq c k\}$. Let $i \neq j$ be two distinct elements in $I$. By the promise
of $\disj_{n, ck, k}$, regardless of whether the input is a NO instance or a YES
instance, there can be at most one index $\ell$ for which $x_{\ell} > 1$. This fact uses that
$c k > 1$, that is, $e \cdot \varepsilon n^{1/p} > 1,$ which holds since $\varepsilon > 1/(e n^{1/p}).$
Consequently, either
$x_i \leq 1$ or $x_j \leq 1$. Without loss of generality, suppose $x_i \leq 1$. Then
\begin{eqnarray}\label{eqn:badCoordinate}
\tilde{x}_i - x_i \geq c k - 1 = (e/4) 4 \varepsilon n^{1/p} - 1 > 2 \varepsilon n^{1/p},
\end{eqnarray}
where the last
inequality holds if $(e-2) \varepsilon n^{1/p} > 1$, that is, $\varepsilon > 1/((e-2) n^{1/p})$. 

Combining (\ref{eqn:pNorm}) and (\ref{eqn:badCoordinate}), 
it follows that $\tilde{x}_i - x_i > \varepsilon \|x\|_p$, which implies that
the streaming algorithm $S$ failed. In each independent 
execution of Step 1, we can assume that the stream
$T^1 \circ T^2 \circ \cdots \circ T^k$ created is the same, and thus the probability $p$ of success
of the streaming algorithm is the same, and $p \geq 2/3$. Thus, $R$ is a geometric random variable
with probability of success $p \geq 2/3$. By Fact \ref{fact:geometric}, $H(R) = O(1)$. 

By definition, $cCost_{\eta_0}(\Pi_S) = I(\Pi_S ; X | D)$, where $D$ is as defined
in distribution $\eta_0$, and we abuse notation and let
$\Pi_S$ denote the transcript of protocol $\Pi_S$. Here $X = (X_1, \ldots, X_k)$, where $X_j$
is the input to Player $j$. By part 3 of Proposition \ref{prop:mut}, 
\begin{eqnarray}\label{eqn:intermediate}
I(\Pi_S ; X | D) \leq I(\Pi_S ; X | D, R) + H(R) \leq I(\Pi_S ; X | D, R) + O(1).
\end{eqnarray}
By definition of conditional mutual information, 
\begin{eqnarray}\label{eqn:main}
I(\Pi_S ; X | D, R) = \sum_{r = 1}^{\infty} I(\Pi_S ; X | D, R = r) (1-p)^{r-1} p.
\end{eqnarray}
We bound each summand $I(\Pi_S ; X | D, R = r)$. Conditioned on $R = r$, the transcript
$\Pi_S$ is equal to $(S^1, \ldots, S^r, W)$, where $S^{\ell} = (S^{\ell, 1}, \ldots, S^{\ell, k})$
and $S^{\ell, j} = S(T^1 \circ T^2 \cdots \circ T^j)$ is the state of the streaming algorithm
$S$ posted to the blackboard by player $j$ in the $\ell$-th execution of Step 1. Here 
$W$ is either equal to the string ``empty", or $W = (X_{1, i}, \ldots, X_{k, i})$, where 
$I = \{i\}$ is the final setting of $I$ by $\Pi_S$. Note that we do not need to explicitly include
the ``start over" messages in the transcript, as they can be inferred given the condition $R = r$. 
We also do not need to include the index $i$ defining $W$ in the transcript, as this can be determined
from the other messages in the transcript. 

By the chain rule for conditional mutual information (part 4 of Proposition \ref{prop:mut}), 
$I(\Pi_S ; X | D, R = r)  =  \left [ \sum_{\ell = 1}^r I(S^{\ell} ; X | S^1, \ldots, S^{\ell-1}, D, R = r) \right ] + I(W ; X | S^1, \ldots, S^{r}, D, R = r)
 \leq  \left [\sum_{\ell = 1}^r H(S^{\ell}) \right ] + H(W | S^1, \ldots, S^{r}, D, R = r) 
 \leq  r \cdot s \cdot k + 1$,
where the first inequality uses part 2 of Proposition \ref{prop:mut}, while the second inequality
uses part 1 of Proposition \ref{prop:mut} together with the fact that conditioned on
$S^1, \ldots, S^{\ell}, D$, and $R = r$, this either fixes $W$ to be the string ``empty", or it fixes
$i$ and $W$ is deterministic given $X_{D, i}$. In either case, 
$H(W | S^1, \ldots, S^{r}, D, R = r) \leq 1$. Since $r s k \geq 1$, we have $I(\Pi_S ; X | D, R = r) \leq 2 r s k$. Plugging into
(\ref{eqn:main}), 
$I(\Pi_S ; X | D, R) \leq \sum_{r = 1}^{\infty} 2r s k \cdot (1/3)^{r-1} (2/3) = O(sk).$
Plugging this into (\ref{eqn:intermediate}), we get
$I(\Pi_S ; X | D) = O(sk) + O(1) = O(s k),$ which completes the proof. 
\end{proof}

\begin{theorem} (Streaming Lower Bound)\label{thm:lower} 
For any constant $p \geq 1$ and any $0 < \varepsilon < 1/4$, any $O(1)$-pass no-underestimation $\ell_p$-point query insertion-only algorithm uses $\Omega(\min(n, n^{1-1/p}/\varepsilon))$ space. 
\end{theorem}
\begin{proof}
Let $S$ be an $O(1)$-pass no-underestimation $\ell_p$-point query algorithm with space $s$. 
Let $\varepsilon' = \max(\varepsilon, 1/((e-2) n^{1/p}))$. 
Note that $S$ is also an $O(1)$-pass no-underestimation $\ell_p$-point query algorithm with space $s$
and parameter $\varepsilon'$. 
Thus, by Lemma \ref{lem:correctness}, 
there is a protocol $\Pi_S$ which solves $\disj_{n, ck, k}$ with probability $1$, 
where $c = e/4$ and $k = 4 \varepsilon n^{1/p}$. 
Consequently, and using $1/((e-2) n^{1/p}) \leq \varepsilon' < 1/4$, by Lemma \ref{lem:cic}, we have 
$cCost_{\eta_0}(\Pi_S) = O(s k).$ We apply Corollary \ref{cor:main} to
conclude $\Omega(n) = cCost_{\eta_0}(\Pi_S) = O(s k),$
and thus, $s = \Omega(n/k) = \Omega(n^{1-1/p}/\varepsilon') = \Omega(\min(n, n^{1-1/p}/\varepsilon))$.  
\end{proof}

\section{Optimal No-Underestimation Sketch}\label{sec:upper}
%

\begin{theorem}\label{thm:upper}
    Let $1 < p < \infty$ be fixed, and suppose there exists some fixed constant $c > 0$ such that $1 \le \varepsilon^{-1} \le n^{(1-c)/p}$ (equivalently, $\varepsilon^{-1} \le n^{1/p - \Omega(1)}$). Then, there exists a randomized sketch of dimension $O(\varepsilon^{-1} \cdot n^{1-1/p})$ on any nonnegative vector $x \in \BR_{\ge 0}^{n}$ that never underestimates any $x_i$, but with probability at least $1-1/n$ does not overestimate any $x_i$ by more than $\varepsilon \cdot \|x\|_p$.
\end{theorem}

\begin{proof}
    We analyze Count-Min sketch with $t=O( 1/(1-1/p))$ hash tables of  size $k= 4 \varepsilon^{-1} \cdot n^{1-1/p}$. The sketch length will be $O(kt) = O(\eps^{-1} \cdot n^{1-1/p})$ for a fixed $p$.
    Let  $r = (2 \varepsilon^{-1} \log n)^p$. We define $\text{head}_r(x)$ to be the set of the $r$ largest (in magnitude) coordinates in $x$. We also define $\text{tail}_r(x) = U \backslash \text{head}_r(x)$. Finally, we define 
    $\|x\|_{p, \text{tail}(r)} = \|x_{\text{tail}_r(x)}\|_p = \left(\sum\limits_{i \in \text{tail}_r(x)} x_i^p\right)^{1/p}.$ 
    For simplicity, we suppose that $\|x\|_{\infty, \text{tail}(r)} = 1$, i.e., we have normalized $x$ so that the $(r+1)$th largest element of $x$ is $1$, unless there are at most $r$ nonzero coordinates of $x$, in which case $\|x\|_{\infty, \text{tail}(r)} = 0$. 
    
 Consider any coordinate $j$. We will analyze the estimation error of that coordinate. First, we observe that, with probability at least $1-r/k$, none of the coordinates in $\text{head}_r(x) -\{i\}$ is hashed to the same bucket as $i$. Conditioning on this event, if there are at most $r$ nonzero coordinates in $x$, then there will be no estimation error. Otherwise, the estimation error of that coordinate is bounded from above by a random variable $X = X_1+\dots+X_n,$ where if $i \in \text{tail}_r(x)$, $X_i$ equals $x_i$ with probability $1/k$ and $0$ otherwise, and if $i \in \text{head}_r(x),$ $X_i = 0$. Then,
$$\BP(X > \varepsilon \cdot \|x\|_p) \le \frac{\BE[e^X]}{e^{\varepsilon \cdot \|x\|_p}} = e^{-\varepsilon \cdot \|x\|_p} \cdot \prod_{i \in \text{tail}_r(x)} \left(1 + \frac{e^{x_i}-1}{k}\right) \le \exp\left(\frac{2}{k} \cdot \|x\|_{1, tail(r)} - \varepsilon \cdot \|x\|_p\right).$$

    Now, we know that $\|x\|_{1, \text{tail}(r)} \le \|x\|_1 \le n^{1-1/p} \cdot \|x\|_p.$ Therefore, $\frac{2}{k} \cdot \|x\|_{1, \text{tail}(r)} \le \frac{\varepsilon}{2} \cdot \|x\|_p,$ so $\BP(X > \varepsilon \cdot \|x\|_p) \le \exp\left(-\varepsilon \|x\|_p/2\right).$ But $\|x\|_p \ge \|x\|_{p, \text{head}(r)} \ge r^{1/p} = 2 \varepsilon^{-1} \log n$ by our normalization of $x$, so $\BP(X > \varepsilon \|x\|_p) \le 1/n$. 
    
    Thus, we have shown the estimation error of a fixed coordinate $j$ exceeds $\varepsilon \|x\|_p$ with probability at most $r/k+1/n = O(\varepsilon^{1-p} (\log n)^p/n^{1-1/p})$. Since $c, p$ are fixed and $\varepsilon \ge n^{-(1-c)/p}$, $\varepsilon^{1-p} \log(n)^p/n^{1-1/p} \le n^{-c \cdot (p-1)/p} \cdot (\log n)^p = n^{-\Omega((p-1)/p)}$. Thus, by using $t=O(p/(p-1)) = O( 1/(1-1/p))$ hash tables, we get that at least one hash table will return an estimate with error at most $\varepsilon \|x\|_p$ with probability at least $1-1/n^2$. By the union bound it follows that all coordinates have error bounded by $\varepsilon \|x\|_p$ with probability at least $1-1/n$.
\end{proof}

\begin{Remark}
    This sketching algorithm implies an $O(\eps^{-1} \cdot n^{1-1/p} \cdot \log n)$-bit streaming algorithm for strict turnstile polynomial-length streams. Although we assume full independence of the hash tables, which requires storing $\Omega(n)$ bits of randomness, one can use the pseudorandom generator of Nisan and Zuckerman \cite{nisanzuckerman1996} to produce a streaming algorithm using space only $O(\log n)$ times the sketch length, since the sketch length, which is $O(k \cdot t)$, is polynomially related to $n$.
\end{Remark}

\begin{Remark}
    In \cite{jw19} upper bounds in the message-passing multiparty communication were studied. Here each of $m$ players holds a nonnegative vector $x^i \in \{0, 1, 2, \ldots, M\}^n$ for some $M = \textrm{poly}(n)$, the players share a common random string, and their goal is to compute a function of their joint inputs by communicating a small number of bits. The players correspond to the nodes of a graph of bounded diameter and communicate along the edges. For the $\ell_p$-point query problem with no underestimation, applied to the vector $\sum_{i=1}^m x^i$, we can improve the $O(\log n)$ bits required to store each coordinate of the sketch of Theorem \ref{thm:upper}, by having each player round each coordinate of its sketched vector up to the nearest power of $(1+\epsilon)$. Thus, when communicating the sketch to another player, it needs only $O(\log \log n + \log(1/\epsilon))$ bits per sketching dimension. Further, the no-underestimation property holds, since counters have only been rounded up, and each counter will be at most $(1+O(\epsilon))$ times its actual value after merging the sketches of all players. The latter property follows since sketches need only be merged $O(1)$ times, using the fact that the diameter is bounded. We note that this is optimal up to an $O(\log \log n + \log(1/\epsilon))$ factor, since the lower bound of Theorem \ref{thm:lower} holds already in the blackboard communication model, where each player posts its message and can be seen by all other players, which is a stronger model than the one in \cite{jw19}. We refer the reader to \cite{jw19} for further details of the model.
\end{Remark}

\section{No-Overestimation}
In this section we give matching sketching lower and upper bounds for no-overestimation algorithms. 

\subsection{No-overestimation lower bound}
Our main technical result is the following theorem.

\begin{theorem} \label{mainNoOver}
    Fix $1 \le T \le n/2.$ Let $v^{(i)}$ be the vector with $i$th coordinate $1$ and $j$th coordinate $\frac{1}{T}$ for all $j \neq i$. Then, there exists an absolute constant $c > 0$ such that for any $k \le c \cdot n/T$ and real-valued matrix $A \in \BR^{k \times n}$, there exists a subset $S \subset [n]$ of size $n/2$ such that for any $i \in S$, there exists $x \in \BR^n$ such that $Ax = Av^{(i)},$ $x$ has only nonnegative entries, and $x_i = 0.$
\end{theorem}

\begin{proof}
    We first prove a weaker version of theorem where we just prove there exists a subset $S \subset [n]$ of size $1$, i.e., there exists some $i \in [n]$ and $x \in \BR_{\ge 0}^n$ such that $Ax = Av^{(i)}$ with $x_i = 0$. However, we relax the assumption $T \le n/2$ to $T \le n$.

    Fix some $1 \le T \le n$. Now, for any fixed matrix $A \in \BR^{k \times n}$ and fixed $1 \le i \le n,$ consider the linear program $\min x_i: Ax = Av^{(i)}, x \ge 0$. The optimal objective being $0$ is equivalent to there existing $x$ such that $A(x - v^{(i)}) = 0$, $x \ge 0,$ and $x_i = 0$. Now, let $B$ be an $n \times (n-k)$ matrix such that the kernel of $A$ equals the image of $B$. Then, this is equivalent to saying there exists $y \in \BR^{n-k}$ such that $By$ can be written as $x- v^{(i)},$ or equivalently, $(By)_i = -1$ and $(By)_j \ge -\frac{1}{T}$. This is equivalent to there being $y \in \BR^{n-k}$ such that $(By)_i \le -1$ and $(By)_j \ge -\frac{1}{T}$ for all $j \neq i$, since if there existed $y$ such that $(By)_i < -1$ and $(By)_j \ge -\frac{1}{T}$, we could scale $y$ by a factor less than $1$ so that $(By)_i = -1,$ and we would still have $(By)_j \ge -\frac{1}{T}$. Finally, if we let $B_{-i}$ be the matrix where the $i$th row is negated, these conditions are equivalent to $(B_{-i} y)_i \ge 1$ and $(B_{-i} y)_j \ge -\frac{1}{T}$.
    
    To summarize, we define $w^{(i)}$ as the vector with $i$th coordinate $1$ and $j$th coordinate $-\frac{1}{T}$ for all $j \neq i$. Then, $\min x_i: Ax = Av^{(i)}, x \ge 0$ having objective $0$ is equivalent to $\min 0 \cdot y: B_{-i} y \ge w^{(i)}$ having objective $0$ (as opposed to $\infty$). The dual linear program of this is $\max w^{(i)} \cdot z: (B_{-i})^T z = 0, z \ge 0$. Therefore, if there does not exist $x$ such that $Ax = Av^{(i)}, x \ge 0,$ and $x_i = 0,$ then there exists $z^{(i)} \ge 0$ such that $(B_{-i})^T z^{(i)} = 0$ and $w^{(i)} \cdot z^{(i)} > 0$. Then, if we let $m^{(i)}$ be the vector which is the same as $z^{(i)}$ but with the $i$th entry negated, then $m^{(i)}_i \le 0,$ $m^{(i)}_j \ge 0$ for all $j \neq i,$ $-m^{(i)}_i > \frac{1}{T} \cdot \sum_{j \neq i} m^{(i)}_j$, and perhaps most importantly, $B^T m^{(i)} = 0$. We can thus scale $m^{(i)}$ so that $m^{(i)}_i = -1$, but then $\sum_j |m^{(i)}_j| \le 1 + T$. Then, if we let $M$ be a matrix with $i$th row $m^{(i)},$ we have that $B^T M = 0$ so $M$ has rank at most $k = rk(A)$, but by Lemma \ref{RankLB2}, the rank of $M$ is $\Omega(n/T).$ Thus, $k = \Omega(n/T)$, or else some $v^{(i)}$ will have its $i$th coordinate estimated as $0$.   
    
    We now prove the full version of the theorem, where we show there is a set $S$ of size $n/2$. Suppose the contrary, and suppose the maximal such set $S$ has size less than $n/2$. In this case, the set $S^c = [n] \backslash S$ has size at least $n/2$, and for every $i \in S^c$ and all $x \in \BR_{\ge 0}^{n}$ with $x_i = 0$, $Ax \neq Av^{(i)}$. Now, let $A_S \in \BR^{|S| \times n}$ be the matrix corresponding to the columns of $A$ in $S$, and let $A_{S^c} \subset \BR^{|S^c| \times n}$ be the matrix corresponding to the columns of $A$ in $S^c$. In addition, for any $i \in S^c,$ we define $v_{S^c}^{(i)}$ to be the $|S^c|$-dimensional vector, indexed by entries in $S^c$, where the entry corresponding to $i$ is a $1$ and all other entries are $1/T$. Now, since $S^c \ge n/2,$ we can apply the weaker version of the theorem so say that there exists $i \in S^c$ and $x' \in \BR^{S^c}$, indexed by elements in $S^c$, such that $A_{S^c}x' = A_{S^c}v_{S^c}^{(i)},$ $x'$ only has nonnegative entries, and $x'_i = 0$. We can apply the theorem since $T \le n/2 \le |S^c|,$ and since $k \le c/2 \cdot n/T$ implies $k \le c \cdot |S^c|/T$ (we may replace $c$ with $c/2$). Therefore, by adding back the remaining entries of $v^{(i)}$ corresponding to entries of $S$ (which are all $1/T$), and making $x \in \BR^n$ such that $x_i = 1/T$ for $i \in S$ and $x_i = x'_i$ for $i \in S^c,$ we have that 
$$Ax = A_{S^c} x' + \frac{1}{T} \sum_{i \in S} A_i = A_{S^c}v_{S^c}^{(i)} + \frac{1}{T} \sum_{i \in S} A_i = A v^{(i)}.$$
    This contradicts the fact that $S$ is maximal, since we can add $i \in S^c$ to it. This proves the full version of the theorem.
\end{proof}

Theorem \ref{mainNoOver} allows us to prove our desired sketching lower bound, which we now state and prove.

\begin{corollary}
\label{cor:noOver}
    A sketching algorithm that returns $\hat{x}$ such that $\hat{x}_i \le x_i$ deterministically but $\hat{x}_i \ge x_i - \varepsilon \cdot ||x||_p$ for all $i \in U$ holds with probability at least $1/2$ must use at least $\Omega(\min(n, \varepsilon^{-1} \cdot n^{1 - 1/p}))$ rows.
\end{corollary}

\begin{proof}
    Suppose $A$ is a sketch matrix, let $T = \max(1, 2 \varepsilon \cdot n^{1/p})$, and suppose the stream has final vector $v^{(i)} = (\frac{1}{T}, \cdots, \frac{1}{T}, 1, \frac{1}{T}, \cdots, \frac{1}{T})$ for a uniformly randomly chosen $i$. The sketch is $Av^{(i)}$, and with probability at least $1/2$ for any fixed $A$ (if $i \in S$ for $S$ as defined in Theorem \ref{mainNoOver}),
    the output $\hat{v}$ must satisfy $\hat{v}_i = 0$ if $k \le c \cdot n/T$, since $Av = Ax$ for some $x$ with $x_i = 0$ and we deterministically cannot overestimate. Also, note that $||v^{(i)}||_p \le 1 + n^{1/p} \cdot \frac{1}{T} \le 1 + \varepsilon^{-1}/2 < \varepsilon^{-1}$, so $\varepsilon \cdot ||v^{(i)}||_p < 1.$ So, for any $i \in S$, if the stream ends with $v^{(i)}$, we do not satisfy the point query lower bound, and therefore $\hat{x}_i < x_i + \varepsilon \cdot ||x||_p$, unless we use $\Omega(n/T) = \Omega(\max(n, \varepsilon^{-1} \cdot n^{1-1/p}))$ rows. This assumes that $A$ is a fixed sketch matrix, but even if $A$ is randomized, the claim still holds because $i$ is chosen randomly, so we still have $i \in S$ with at least $1/2$ probability where $S$ is a set of size at least $n/2$ that may depend on $A$.
\end{proof}

\subsection{Near-optimal sketch that does not overestimate}

In this section we provide a sketching method of dimension almost matching the lower bound in Corollary~\ref{cor:noOver}, up to a factor of $\log n \cdot \log \varepsilon^{-1}$. Specifically, we show the following theorem.

\begin{theorem} \label{thm:alg_no_overestimation}
    Let $1 \le p < \infty$. Then, there exists a randomized sketch of dimension $O(p \cdot \varepsilon^{-1} \log \varepsilon^{-1} \cdot n^{1-1/p} \log n)$ on any nonnegative vector $x \in \BR_{\ge 0}^{n}$ that never overestimates any $x_i$, but with probability at least $9/10$ does not underestimate any $x_i$ by more than $\varepsilon \cdot \|x\|_p$.
\end{theorem}

One important tool we will use in establishing this theorem is the following theorem on \emph{deterministic} point query.

\begin{theorem} \label{thm:deterministic_no_overestimation} \cite{nelson2014deterministic}
    There exists a deterministic sketching algorithm which creates a sketch of length $O(\varepsilon^{-2} \log n)$ that, for any $x \in \BR^n,$ can produce a vector $\hat{x}$ such that $\|x-\hat{x}\|_\infty \le \varepsilon \cdot \|x_{-\lfloor 1/\varepsilon^2 \rfloor}\|_1,$ where $x_{-\lfloor 1/\varepsilon^2 \rfloor}$ is the vector with the $\lfloor 1/\varepsilon^2 \rfloor$ largest entries (in absolute value) removed.
\end{theorem}

While this result may not seem sufficient to get a nearly linear dependence in the sketch length of $\varepsilon^{-1}$ (in fact, the $\varepsilon^{-2}$ dependence is known to be necessary), we in fact avoid this issue by only applying theorem \ref{thm:deterministic_no_overestimation} when $\varepsilon$ is a constant.
Indeed, as a direct corollary, we have the following result.

\begin{corollary} \label{cor:deterministic_no_overestimation}
    There exists a deterministic sketching algorithm which creates a sketch of length $O(\log n)$ that, for any $x \in \BR_{\ge 0}^n$, creates a vector $\hat{x}$ such that for all $i \in [n]$, $x_i - \frac{2}{9} \cdot (\|x\|_1-x_i) \le \hat{x}_i \le x_i$.
\end{corollary}

\begin{proof}
    Note that $\|x_{-100}\|_1 \le \|x\|_1 - x_i$ for any $i$. Thus, by Theorem \ref{thm:deterministic_no_overestimation}, there exists a deterministic $O(\log n)$-length sketch that can find a vector $\tilde{x}$ such that 
    $$1.1 x_i - 0.1 \|x\|_1 = x_i - 0.1 (\|x\|_1-x_i) \le \tilde{x}_i \le x_i + 0.1 (\|x\|_1-x_i) = 0.9 x_i + 0.1 \|x\|_1.$$
    By increasing the length of the sketch by $1$, we can also keep track of $\|x\|_1 = \sum_{i = 1}^{n} x_i,$ since the coordinates of $x$ are nonnegative. Therefore, by letting $\hat{x}_i = \frac{10}{9} \cdot \left(\tilde{x}_i - 0.1  \|x\|_1 \right),$ we have that for all $i \in [n]$,
\begin{equation} \label{eq:overestimation}
    \frac{10}{9}\left(1.1 x_i - 0.1 \|x\|_1 - 0.1 \|x\|_1\right) \le \hat{x}_i \le \frac{10}{9} \cdot \left(0.9 x_i + 0.1 \|x\|_1 - 0.1  \|x\|_1\right).
\end{equation}
    However, we can lower bound the left hand side of Equation \eqref{eq:overestimation} by $x_i - \frac{2}{9} \left(\|x\|_1 - x_i\right)$, and simplify the right hand side of Equation \eqref{eq:overestimation} as $x_i$. Overall, we have that $x_i - \frac{2}{9} \cdot (\|x\|_1-x_i) \le \hat{x}_i \le x_i$, as desired.
\end{proof}

We are now ready to prove Theorem \ref{thm:alg_no_overestimation}. Our sketching algorithm works as a combination of the deterministic sketch based on Corollary \ref{cor:deterministic_no_overestimation} along with hashing-based ideas similar to Count-Min.

\begin{proof}[Proof of Theorem \ref{thm:alg_no_overestimation}]
    First, we hash $U = [n]$ uniformly into $k = 4 \varepsilon^{-1} \cdot n^{1-1/p}$ buckets $B_1, B_2, \dots, B_k$. For the set of indices mapped to some bucket $B_j$, we use Corollary \ref{cor:deterministic_no_overestimation} to estimate $x_i$ for each $i \in B_j$. Indeed, we obtain an estimate $\hat{x}_i$ for each $i \in [n]$ such that for all $i$, if $b(i)$ is index of the bucket that $i$ is mapped to, then deterministically,
$$x_i - \frac{2}{9} \cdot \sum_{\substack{i' \neq i \\b(i) = b(i')}} x_{i'} \le \hat{x}_i \le x_i.$$
    In addition, we know that for any fixed $i$, the expectation of the sum of $x_{i'}$ over $b(i) = b(i'), i' \neq i$ is at most $\frac{1}{k} \cdot \|x\|_1 \le \frac{\varepsilon}{4} \cdot \|x\|_p$, since each $x_{i'}$ is hashed to the same bucket as $x_i$ with $1/k$ probability and $\|x\|_1 \le \|x\|_p \cdot n^{1-1/p}$. So, for any fixed $i$, with probability at least $3/4,$ $x_i - \frac{2}{9} \cdot \varepsilon \cdot \|x\|_p \le \hat{x}_i$, and deterministically, $\hat{x}_i \le x_i$.

    Our final estimate will be to run $t = O(p \cdot \log \varepsilon^{-1})$ independent copies of this algorithm. If $\hat{x}_i^{(\ell)}$ represents the estimate by the $\ell$th copy of this algorithm, our final estimate for each $x_i$ will be $\bar{x}_i = \max(0, \max_{1 \le \ell \le t} \hat{x}_i^{(\ell)}).$ Note that $\bar{x}_i$ is never an overestimate, and for any $i$ such that $x_i \le \varepsilon \cdot \|x\|_p$, $\bar{x}_i \ge 0$, so with probability $1$, it does not underestimate by more than $\varepsilon \cdot \|x\|_p$. Finally, if we define $S \subset [n]$ to be the set of indices $i$ such that $\hat{x}_i \ge \varepsilon \cdot \|x\|_p$, it is immediate that $|S| \le \varepsilon^{-p}$. For each $i \in S,$ we know that $x_i - \frac{2}{9} \cdot \varepsilon \cdot \|x\|_p \le \bar{x}_i$ with probability at least $1-(3/4)^{t} \le 1 - \varepsilon^p/10,$ which means by a union bound over $i \in S$, we have that $\bar{x}_i$ does not underestimate by more than $\varepsilon \cdot \|x\|_p$ for all $i \in S$ with probability at least $9/10$.

    The length of the sketch is $O(k \cdot t \cdot \log n) = O(p \cdot \varepsilon^{-1} \log \varepsilon^{-1} \cdot n^{1-1/p} \log n)$, since we hash $t$  times into $k$ buckets, and use a sketch of length $O(\log n)$ on each one.
\end{proof}

\begin{Remark}
    We note that this procedure is indeed a \emph{sketching algorithm} for the same reason that Count-Min sketch is, and since the composition of two linear sketches is a linear sketch. 
\end{Remark}

\paragraph{Acknowledgments:}  This research was supported in part by the NSF TRIPODS program (awards CCF-1740751 and DMS-2022448), NSF awards CCF-2006798 and CCF-1815840, Office of Naval Research grant N00014-18-1-2562, Simons Investigator Awards, and an NSF Graduate Fellowship. 

\newcommand{\etalchar}[1]{$^{#1}$}

\section{Appendix} 

\begin{proof}(of Lemma \ref{RankLB1}): 
    Note that if $R$ is symmetric, this is trivial since $tr(R) = \sum \lambda_i(R)$ and $||R||_F^2 = \sum \lambda_i^2(R).$ Now, for an arbitrary $R$, let $S = (R+R^T)/2.$ Then, $rk(S) \ge tr(S)^2/||S||_F^2.$ However, it is clear that $rk(S) \le 2 \cdot rk(R)$, $tr(S) = tr(R),$ and $||S||_F^2 \le ||R||_F^2.$ Therefore, $2 \cdot rk(R) \ge rk(S) \ge tr(S)^2/||S||_F^2 \ge tr(R)^2/||R||_F^2$.
\end{proof}

\begin{proof}(of Lemma \ref{RankLB2}):
    First, we pick the following set of elements of $M$. Let $i_1, j_1$ be such that $|M_{i_1, j_1}|$ is maximized (if there is a tie, choose any maximal $i_1, j_1$). Then, for all $2 \le k \le n,$ in that order, we choose $i_k \in [n] \backslash \{i_1, \dots, i_{k-1}\}$ and $j_k \in [n] \backslash \{j_1, \dots, j_{k-1}\}$ that maximizes $|M_{i_k j_k}|$. This creates two permutations $i_1, i_2, \dots, i_n$ and $j_1, j_2, \dots, j_n$ of $[n]$.

    Now, let $r$ be some power of $2$ between $1$ and $T$ and consider the interval $T \subset [n]$ such that $t \in T$ if $|M_{i_t, j_t}| \in [r, 2r].$ Let $c_r$ be the size of this interval. Now, consider the matrix restricted to rows $i_t$ for $t \in T$ and columns $j_t$ for $t \in T$. Then, all of the diagonal terms (where $M_{i_tj_t}$ are the new diagonals) are between $r$ and $2r$ in magnitude, and all of the terms in the matrix are also at most $2r$ (or else we would have found a different maximum). Call this restricted matrix $R$, and let $R'$ be the matrix created when each row of $R$ is either preserved or negated so that the diagonal entries are all positive. Then, the diagonal entries are between $r$ and $2r$, which means that $tr(R') = \Theta(r \cdot c_r)$. Also, since all entries of $R'$ are at most $2r$ in magnitude and the sum of the absolute values of the entries of each row is at most $T,$ we have that $||R'||_F^2 = O(r \cdot T \cdot c_r).$ Therefore, by Lemma \ref{RankLB1}, $rk(R') \gtrsim (r \cdot c_r)^2/(r \cdot T \cdot c_r) = r c_r/T,$ and it is clear that $rk(R) = rk(R')$.

    Now observe that there exists some $r = 2^i$, $0 \le i \le \lfloor\log_2 T\rfloor$, such that $r c_r \ge 0.25 n.$ If not, then for all $r = 2^i$ with $0 \le i \le \lfloor \log_2 T \rfloor$, we have that $c_{2^i} < 0.25 n/2^i$. Adding these together we get that $\sum_{r = 2^i: i = 0}^{\lfloor \log_2 T \rfloor} c_r < 0.5 n.$ However, we have $\sum_{r = 2^i: i = 0}^{\lfloor \log_2 T \rfloor} c_r \ge 0.5 n.$ This is because (by the assumption) all of the diagonal entries of $M$ are at least $1$, which means that after any $n' < 0.5 n$ steps of picking $(i, j),$ there must be some diagonal entry left in the submatrix. Therefore, the number of $k$ such that $|M_{i_k, j_k}| \ge 1$ is at least $0.5$, so $\sum_{r = 2^i: i = 0}^{\lfloor \log_2 T \rfloor} c_r \ge 0.5 n$. Therefore, there is some $r$ such that $r c_r/T \ge 0.25 \cdot n/T$, so there is some submatrix $R$ of $M$ such that $rk(R) = \Omega(n/T).$ Since $rk(M) \ge rk(R),$ this concludes the proof.
\end{proof}

\end{document}

%% file: intro.tex
Frequency estimation, also known as the Point Query problem, is one of the most fundamental problems in streaming algorithms.  Given a stream $S$ of elements from some universe $U=[n]=\{1 \ldots n\}$, the goal is to compute a short ``sketch'' of $S$ so that for any element $i \in U$, one can estimate the number $x^S_i$ of times $i$ occurs in $S$ based on the sketch alone. 
Furthermore, the computation should be performed in a ``streaming'' fashion, by performing only one pass (or few passes) over the data. 
Over the last two decades, dozens of algorithms for this problem have been developed. Some of them, such as Count-Min~\cite{cormode2005improved} and Count-Sketch~\cite{charikar2002finding}, have found applications in multiple areas, including machine learning, natural language processing, network monitoring and security, and have been implemented in popular data processing libraries, such as Algebird and DataSketches. See~\cite{cormode2020small}, sections 3.4 and 3.5, for further discussion of applications.

Both Count-Min and Count-Sketch are {\em linear} sketches. Specifically, the algorithms compute a vector $Ax^S$, where $x^S$ is the frequency vector for the stream $S$, and $A$ is the sketch matrix defined by the respective algorithm.  The linearity has multiple benefits. First,  the data elements can be deleted as well as inserted\footnote{For Count-Min, this holds under the condition that $x^S \ge 0$ at the end of the stream - this is often referred to as the {\em strict turnstile model}.}. Furthermore, the sketch is {\em mergeable}~\cite{agarwal2013mergeable}: given two data streams $S$ and $S'$, the sketch of the concatenation of $S$ and $S'$ is equal to the sum of sketches for $S$ and $S'$, i.e.,  $A x^{S \circ S'}= A (x^{S} + x^{S'} ) = A x^{S} + Ax^{S'}$.  The linearity of sketches is also crucial for other applications such as {\em compressed sensing}~\cite{candes2006robust, donoho2006compressed, gilbert2010sparse}. 

The guarantees offered by Count-Min and Count-Sketch are incomparable. The frequency estimator $\tilde{x}$ produced by Count-Min sketch, using $O(1/\varepsilon \cdot \log n)$ dimensions, guarantees that (i) $\|\tilde{x} -x\|_{\infty} \le \varepsilon \|x\|_1$ with high probability, and (ii) $\tilde{x} \ge x$ holds always. Also, Count-Min works under the assumption that $x \ge 0$.  In contrast, Count-Sketch uses $O(1/\varepsilon^2 \cdot \log n)$ dimensions and guarantees that $\|\tilde{x} -x\|_{\infty} \le \varepsilon \|x\|_2$ with high probability. The $\ell_2$ norm is typically smaller (and never greater) than the $\ell_1$ norm, so for constant $\varepsilon$ the error guarantee of Count-Sketch is stronger than the error guarantee of Count-Min. However, Count-Min has the additional ``no underestimation'' property (ii), which is quite useful in applications. In particular, if the goal is to identify all elements $i$ such that $x_i \ge T$ for some threshold $T$, the no-underestimation property guarantees that the algorithm reports all such ``heavy'' elements, i.e., the algorithm has no false negatives. Since in many applications, such as traffic monitoring,  heavy elements indicate the presence of anomalies that need to be investigated further, preventing false negatives  is of paramount importance. In contrast, Count-Sketch can suffer from both false negatives and false positives, although the probability of either can be made arbitrarily small by increasing the sketch size. We also note that it is possible to use Count-Min to obtain a non-trivial sketch with $\ell_2$ error guarantee {\em and} no-underestimation property by setting  $\varepsilon$ equal to $1/\sqrt{n}$, which ensures that the error is at most $\|x\|_1/\sqrt{n} \le \|x\|_2$.  However, this sketch uses $O(\sqrt{n} \log n)$ dimensions, which is much larger than the space bound of Count-Sketch. Furthermore, one can simulate Count-Sketch using Count-Min sketch by doubling the sketch dimension (\cite{cormode2020small}, section 3.5), but the resulting estimate only satisfies the guarantees of Count-Sketch; in particular, it does not have the ``no underestimation'' property.

This state of affairs leads to a natural question: is it possible to design the ``best of both worlds'' sketching or streaming algorithm, with error guarantees depending on the $\ell_2$ norm and space comparable to Count-Sketch, but (like Count-Min) also has the no-underestimation property?  

\paragraph{Our results:} Our main contributions show that the answer to the above question is negative.  We consider this problem in two models: linear sketching and streaming algorithms. We show\footnote{The following statements assume that $\varepsilon$ is not ``too small'' as a function of $n$. Please see the relevant sections for the complete result statements.} that:
\begin{itemize} 
\item Any {\em linear sketch} satisfying the $\ell_2$ error guarantee  with probability at least $2/3$ and having the no-underestimation property must be of {\em dimension} at least $\Omega(\sqrt{n}/\varepsilon)$, even if the sketched vectors are non-negative. The result can be generalized to any $\ell_p$ norm, yielding a dimension lower bound of $\Omega(n^{1-1/p} /\varepsilon)$. 
\item Any {\em streaming algorithm}  satisfying the $\ell_p$ error guarantee  with probability at least $2/3$ and having the no-underestimation property must use at least  $\Omega(n^{1-1/p} /\varepsilon)$ {\em bits}. This holds even for algorithms that only allow insertions (not deletions), and that are allowed $O(1)$ passes.
\end{itemize}

We complement these lower bounds by showing that they are (almost) tight. Specifically, for any fixed $p > 1$, we provide a linear sketch of dimension $O(n^{1-1/p}/\varepsilon)$ that works for non-negative vectors satisfying the above properties. This matches  our lower bound for any fixed $p>1$. The sketch is obtained by refining the analysis of Count-Min for $\ell_p$ norms with $p>1$, improving (by a logarithmic factor in $n$) over the na\"ive bound sketched above. Furthermore, for insertion-only streams, we can implement this algorithm using $O(\log n)$ bit counters per dimension (as long as the stream length is at most polynomial in $n$), 
which yields an $O(n^{1-1/p}/\varepsilon \cdot  \log n)$ bit space bound for fixed $p > 1$. Therefore, our streaming lower bound is tight up to a factor of $ \log n$. In the message-passing multi-party communication model, where  streaming problems were studied recently in \cite{jw19}, this $\log n$ factor can be improved and we obtain tight bounds up to an $O(\log \log n + \log(1/\epsilon))$ factor; see Section \ref{sec:upper} for details. 

We note that our two lower bounds are incomparable. On the one hand, any linear sketch also yields a streaming algorithm, so a streaming lower bound can in principle be used to derive a lower bound on linear sketches as well. On the other hand,  the entries of sketching matrices are real numbers with an arbitrary or even unbounded precision, so translating streaming lower bounds into sketching lower bounds induces an overhead that depends on the precision. 
Furthermore,  our sketching  result lower bounds  the number of {\em dimensions}, while the streaming result lower bounds the number of {\em bits}. Thus, our sketching lower bound is tight, while the streaming lower bound is tight up to a logarithmic factor.

Finally, we study the complementary problem, where the sketch is required to not {\em over}-estimate, i.e., $\tilde{x} \le x$ should hold always. We show that any linear sketch satisfying this property and having the $\ell_p$ error guarantee with probability at least $2/3$ must be of dimension at least $\Omega(n^{1-1/p}/\varepsilon )$, even for $x \ge 0$. We also show this is tight up to polylogarithmic factors in $n$, by giving an appropriate linear sketch. 

\subsection{Related work} To the best our knowledge, the closest prior work is the paper~\cite{brody2011streaming}. It considers streaming algorithms for a collection of problems, including norm estimation and heavy hitters, and shows that when such algorithms are required to output a number that does not underestimate (or overestimate) the true value, then such algorithms must use linear space $\Omega(n)$, unless the error is measured in $\ell_p$ norm for $p = 1$. Their lower bound for the heavy hitters problem (motivated, as in our case, by the Count-Min algorithm), is particularly relevant to the results in this paper. In the heavy hitters problem, the goal is to (i) report all elements $i$ such that $x_i \ge \phi \|x\|_p$, while (ii) not reporting any elements $i'$ for which $x_{i'} < \phi' \|x\|_p$. Here,  $0<\phi' < \phi<1$ are constants that depend on $p$ but not on $n$. The paper shows that if either condition (i) or condition (ii) holds with probability 1, then any streaming algorithm must use $\Omega(n)$ space, even for insertions-only streams, unless $p=1$. 

On the surface, their result might appear to be stronger than our lower bound, and to contradict our upper bound (which is approximately $O(\sqrt{n})$ for $p=2$). This, however, is not the case, because the definition of ``heavy hitters'' in~\cite{brody2011streaming} is {\em relative} with respect to the total norm $\|x\|_p$. This means that estimating the frequencies by itself (as in Count-Min) is not sufficient to identify heavy elements, and one must also estimate the {\em norm} of the stream to be able to determine which elements have estimates exceeding the threshold of $\phi \|x\|_p$. Indeed, the heavy hitter lower bound in~\cite{brody2011streaming} crucially relies on the hardness of estimating the norm with one-sided error. In contrast, our lower bounds apply directly to the frequency estimation problem.

More broadly, communication complexity protocols with one-sided error have been studied extensively in communication complexity. Indeed, the aforementioned lower bound of \cite{brody2011streaming} relies on lower bounds for one-sided-error communication protocols from the seminal paper of~\cite{buhrman1998quantum}. See~\cite{brody2011streaming} for a detailed overview of this line of research. 

{\bf Deterministic streaming algorithms:} Streaming algorithms with one-sided error are generalizations of deterministic streaming algorithms. In the context of frequency estimation and related problems, deterministic algorithms have been studied e.g., in \cite{nelson2014deterministic,li2018deterministic}. However, those algorithms work only for the $\ell_1$ norm. For the $\ell_2$ norm, a recent paper ~\cite{KPW21} showed an $\Omega(\sqrt{n}/\eps)$ lower bound for the deterministic heavy hitters problem with thresholds $\phi=\eps$ and $\phi'=\eps/2$. Our streaming lower bound is a strengthening of that result, showing that the lower bound holds already for algorithms which can be randomized, provided they have one-sided error.

\subsection{Preliminaries}

{\bf Notation:} We will use $[n]$ to denote the set $\{1 \ldots n\}$. As stated in the introduction, we use $x^S$ to denote the frequency vector induced by the stream $S$, i.e., $x^S_i$ is the number of times $i$ appears in $S$. We will often drop the superscript when $S$ is clear from the context. 

For a vector $x \in \BR^n$, we denote it's $\ell_p$ norm as $\|x\|_p = \left(\sum_{i = 1}^{n} |x_i|^p\right)^{1/p}.$ For a real-valued matrix $M$, we let $M^T$ denote the transpose of $M$, $rk(M)$ denote the rank of $M$, and if $M$ is square, we let $tr(M)$ denote the trace of $M$. Finally, we define the Frobenius norm $\|M\|_F = \sqrt{M^T M} = \sqrt{\sum M_{ij}^2}.$

{\bf Count-Min Sketch:}  The sketch is formed by creating $t$ distinct hash functions $h_{\ell}:U \rightarrow [k]$ and $t$ arrays $C_{\ell}$ of size $k$ each. The total space used is of size $tk$.
The algorithm computes $C_{\ell}$'s such that at the end of the stream we have
$C_{\ell}[b] = \sum_{i: h_{\ell}(i)=b} x_i$ for each $b \in [k]$.
For each $i \in U$, the frequency estimate $\tilde{x}_i$ is equal to $\min_{\ell} C_{\ell}[ h_{\ell}(i)]$. 
Note $\tilde{x}_i \ge x_i$ holds always as long as $x \ge 0$. In the context of streaming algorithms we refer to the latter assumption as the ``strict turnstile model''. 
We also note that the mapping $A: \BR^n \to \BR^{kt}$ that maps $x$ to $C_1 \ldots C_t$ as defined above is linear.

%% file: overview.tex
\paragraph{Sketching lower bounds:}
The existence of a sketching algorithm implies there is a family of $k \times n$-dimensional matrices $\mathcal{A}$ such that given $A$ and $Ax$ for $A \sim \mathcal{A}$, one can recover an approximation $\hat{x}$ of $x$. For all $i \in [n]$, this approximation satisfies $\hat{x}_i \ge x_i$  (in the no-underestimation case) or $\hat{x}_i \le x_i$  (in the no-overestimation case), and with probability at least $2/3$, $|\hat{x}_i-x_i| \le \varepsilon \cdot \|x\|_p$ for all $i \in [n]$. Our goal is to show that such a guarantee is not possible unless $k = \Omega(\varepsilon^{-1} \cdot n^{1-1/p})$. 
It is well-known that one can replace the randomness on the matrix $A$ with randomness over the vector $x$ that we wish to estimate. In other words, if such a guarantee is possible, then for any distribution $\mathcal{D}$ over vectors $x \in \BR^n$, there must exist a fixed matrix $\mathcal{A}$ such that given $Ax$, where $x \sim \mathcal{D},$ we can recover a good estimate $\hat{x}$ of $x$. But due to our deterministic no underestimation/overestimation assumption, we have that for any $x$ (possibly not drawn from $\mathcal{D}$), $\hat{x}_i \ge x_i$.

In the case of no-underestimation, the distribution $\mathcal{D}$ is simply a point mass distribution over $v = (1/n, 1/n, \dots, 1/n) \in \BR^n$. We show that for any matrix $A$, there exists a nonnegative vector $x$ such that $Av = Ax$ but some component $x_i$ is large. Because of our deterministic no-underestimation assumption, we cannot underestimate $x_i$, so we will grossly overestimate $v_i$. In the case of no-overestimation, we instead consider the distribution over vectors where all coordinates are $1/T$ for some properly chosen $T$, except one uniformly random coordinate $i$ is chosen to be $1$ (call this vector $v^{(i)}$). This time, we show that for the majority of these vectors, there is some nonnegative $x$ such that $Ax = Av^{(i)}$, but $x_i = 0$. Because of the deterministic no-overesimation assumption, we cannot underestimate $x_i$, so we grossly underestimate $v^{(i)}_i$ for the majority of the vectors $v^{(i)}$.

To prove that these vectors $x$ exist, we write the claims that $Av = Ax,$ $x_i$ is large (or is $0$), and $x$ is nonnegative as a linear program and consider the dual linear program. In both cases, we reduce to two similar matrix inequalities, both of which can be captured by showing that if a matrix has diagonal entries at least $1$ but the sum of the absolute values of each row is less than some $T$, then the rank must be at least $\Omega(n/T)$. This generalizes a well-known fact for $T = 1$, which essentially states that diagonally dominant matrices are invertible \cite{taussky49diagdom}. The rough intuition for the more general matrix inequality comes from assuming that the largest entry in each row is the diagonal entry. In this case, we can bound the sum of squares of each row by $T$, meaning that the Frobenius norm of $M$ is at most $T \cdot n$, whereas the trace is $n$. Some simple inequalities relating the Frobenius norm, trace, and rank are sufficient to establish an $\Omega(n/T)$-rank lower bound. While we cannot assume that the largest entry in each row is the diagonal entry, we show how to swap rows and columns accordingly without affecting the rank, and make the diagonal entries the largest entry in many of the rows, or at least if we restrict to certain submatrices. This idea will allow us to prove the general rank lower bound.

{\bf Streaming Lower bound:}  Our result is based on the lower bound for a $k$-player communication promise problem called ``Mostly Set-Disjointness'', introduced in~\cite{KPW21}. In this problem, denoted by $\disj_{n, l, k}$,  each Player $j \in [k]$ receives an 
	$n$-dimensional input vector $X_j = (X_{j,1}, \dotsc, X_{j,n})$ where 
	$X_{j,i}\in \{ 0,1\}$. 
The input to the protocol falls into either of 
	the following cases. In the NO case, we have that the sets represented by $X_j$'s are disjoint. 
	In the YES case,  all sets are disjoint except for a  a unique element $i\in[n]$ which occurs in exactly $l$ sets.
The paper \cite{KPW21} shows that, for any fixed $0<c<1$, the conditional information complexity of $\disj_{n, cn, k}$ (for an appropriately defined distribution over the inputs) is $\Omega(n)$.

The starting point of the proof is a standard reduction between streaming algorithms and multiparty protocols, where each of the $k$ players in  $\disj_{n, ck, k}$ creates a local stream and  runs a streaming algorithm  on the concatenation of its local streams. For an $O(1)$-pass streaming algorithm, this requires $O(1)$ rounds, where each player speaks once in a round. The last player then looks at the output of the streaming algorithm, and decides on
a course of action. With probability at least $2/3$ and using at most one additional round and $O(k)$ additional bits of communication with each of the players, Player $k$ declares an answer to the $\disj_{n, ck, k}$ problem. With the remaining probability of at most $1/3$, 
Player $k$ asks all players to run an independent execution of the streaming algorithm again on their local streams. This process repeats until Player $k$ finally declares an answer to the $\disj_{n, ck, k}$ problem. 

An important fact is that when Player $k$ declares an answer to the $\disj_{n, ck, k}$ problem, the output is correct with probability $1$. We show this using the properties of a no-underestimation $\ell_p$-point query algorithm. However, it may take many independent executions of the streaming algorithm until Player $k$ declares an answer. We can terminate the protocol after $O(k)$ independent executions, incurring a probability of error for solving $\disj_{n, ck, k}$ of at most $\delta = \exp(-\Theta(k))$, but the total communication for solving $\disj_{n, ck, k}$ will now be $O(k^2 s)$, where $s$ is the space complexity of our streaming algorithm. Indeed, in each of $k$ rounds, each of $k$ players passes the state of the streaming algorithm (or more precisely, posts it to the blackboard) to the next player $O(1)$ times (since the streaming algorithm uses $O(1)$ passes). As the randomized communication complexity of $\disj_{n, ck,k}$ is $\Theta(n)$, this gives us an $s = \Omega(n/k^2)$ lower bound. For the important setting of $p = 2$ and $\varepsilon = \Theta(1)$, we would set $k = \Theta(\sqrt{n})$ and only obtain a trivial $\Omega(1)$ lower bound. 

The main insight is instead to argue that although there may be up to $O(k)$ rounds of communication, the expected number of rounds is $O(1)$, and consequently since the conditional information cost is a quantity {\it measured in expectation}, the additional rounds
do not degrade the lower bound, since they occur with geometrically decreasing probabilities. 

{\bf Sketching Algorithms:} We start with the no-underestimation algorithm. In standard Count-Min sketch with the $\ell_1$-guarantee, one can show that for any $x_i$, the sum of the other coordinates that are hashed to the same bucket as $x_i$ does not exceed $\varepsilon \cdot \|x\|_1$ with at least $2/3$ probability. Count-Min repeats this hashing procedure $O(\log n)$ times, which ensures high accuracy for every $x_i$. Since $\varepsilon/n^{1-1/p} \cdot \|x\|_p \ge \varepsilon \cdot \|x\|_1$, one can apply the same Count-Min sketch with a hash table of size $\varepsilon^{-1} \cdot n^{1-1/p}$ instead of size $\varepsilon^{-1}$. However, if $\varepsilon/n^{1-1/p} \cdot\|x\|_p$ and $\varepsilon \cdot \|x\|_1$ were actually close in value, this would mean that most of the coordinates of $x$ were  small in absolute value, which means we can actually obtain strong concentration of the sum of other items hashed to the same bucket as $x_i$. Indeed, since there are $\varepsilon^{-1} \cdot n^{1-1/p}$ buckets now, for any fixed $i$, with very small failure probability no heavy hitter will be mapped to the same bucket as $i$, and we can use a Chernoff-type argument to bound the sum of remaining items with $1/\text{poly}(n)$ failure probability. Overall, our final algorithm is very similar to Count-Min, but we only need $O(1)$ copies of the hash table as opposed to $O(\log n)$ copies. This will allow us to obtain an $O(\varepsilon^{-1} \cdot n^{1-1/p})$-length sketch.

For the no-overestimation algorithm, we combine the ideas of Count-Min with known deterministic $\ell_1$-point query algorithms, which are able to operate with sketch length $O(\varepsilon^{-2} \log n)$ \cite{nelson2014deterministic}. Since the point query algorithms are deterministic, they are not hard to modify so that they never overestimate. Unfortunately, we cannot directly use this, as for the $\ell_p$-point query, we would only get a $\tilde{O}(\varepsilon^{-2} \cdot n^{2-2/p})$-length sketch, since in the worst case, $\varepsilon/n^{1-1/p} \cdot \|x\|_1 = \varepsilon \cdot \|x\|_p.$ Instead, we will only apply the deterministic point query algorithm for $\varepsilon$ a small constant. Our approach, roughly, is to use a Count-Min-type hashing to split the stream into buckets. We can keep track of the total mass of each bucket, and run a deterministic point-query algorithm on the bucket so that every $x_i$ is not overestimated, but not underestimated by more than $0.1 \cdot \|x^B\|_1,$ where $x^B$ is the substream generated by the elements mapped to the same bucket $B$ that $i$ is mapped to. By running a small number of copies of this procedure, we can ensure that for each $i$, the bucket containing $i$ has small norm in one of the copies, which will be sufficient for our sketch. Our streaming algorithms are obtained by implementing the aforementioned sketching methods.